\documentclass{article}

\usepackage{fullpage}
\usepackage{amsthm}
\usepackage{amsmath}
\usepackage{amssymb}
\usepackage{thmtools,thm-restate}
\usepackage{cite}
\usepackage{tikz}
\usepackage[basic]{complexity}
\usepackage[T1]{fontenc}
\usepackage{enumerate}
\usepackage{paralist}
\usepackage{xspace}
\usepackage{macros}
\newtheorem{definition}{Definition}
\newtheorem{theorem}{Theorem}
\newtheorem{lemma}{Lemma}
\newtheorem{corollary}{Corollary}

\newtheorem{proposition}{Proposition}

\newtheorem{example}{Example}

\title{Looking at Mean-Payoff through Foggy Windows}
\author{
	Paul Hunter\thanks{Authors supported by the ERC inVEST (279499)
	project.} \and
	Guillermo A. P\'{e}rez\thanks{Author supported by F.R.S.-FNRS
	fellowship.} \and
	Jean-Fran\c{c}ois Raskin$^\ast$\\
	\texttt{\{phunter,gperezme,jraskin\}@ulb.ac.be}\\
Universit\'{e} Libre de Bruxelles -- Brussels, Belgium}

\begin{document}

\maketitle

\newcommand{\mpinf}{\underline{\textsf{MP}}}
\newcommand{\mpsup}{\overline{\textsf{MP}}}
\newcommand{\lmax}{\ensuremath{ {l_{\text{max}}} }}
\newcommand{\mpobjinf}{\textsf{MPInf}}
\newcommand{\mpobjsup}{\textsf{MPSup}}
\newcommand{\gw}{\textsf{GW}}
\newcommand{\dirfix}{\textsf{DirFix}}
\newcommand{\udirbnd}{\textsf{UDirBnd}}
\newcommand{\dirbnd}{\textsf{DirBnd}}
\newcommand{\fix}{\textsf{Fix}}
\newcommand{\ufix}{\textsf{UFix}}
\newcommand{\ubnd}{\textsf{UBnd}}
\newcommand{\bnd}{\textsf{Bnd}}

\begin{abstract}
	Mean-payoff games (MPGs) are infinite duration two-player zero-sum games
	played on weighted graphs. Under the hypothesis of perfect information,
	they admit memoryless optimal strategies for both players and can be
	solved in $\NP \cap \coNP$. MPGs are suitable quantitative models for
	open reactive systems. However, in this context the assumption of
	perfect information is not always realistic. For the partial-observation
	case, the problem that asks if the first player has an observation-based
	winning strategy that enforces a given threshold on the mean-payoff, is
	undecidable.
	In this paper, we study the window mean-payoff objectives that were
	introduced recently as an alternative to the classical mean-payoff
	objectives.  We show that, in sharp contrast to the classical
	mean-payoff objectives, some of the window mean-payoff objectives are
	decidable in games with partial-observation.
\end{abstract}

\section{Introduction}

\emph{Mean-payoff games} (MPGs)~\cite{em79} are infinite duration, two-player,
zero-sum games played on weighted graphs, useful for modelling reactive systems
with quantitative objectives and designing algorithms to synthesize controllers
for such systems~\cite{ChakrabartiAHS03}. 
Like other verification games played
on graphs, two players move a token around the graph for an infinite number of
steps. One of the players selects a label, after which the second chooses an
edge with this label. The token is then moved along the selected edge.
This infinite
interaction between the two players results in an infinite path in the graph.
The objective of Player 1 is to maximize the limiting average payoff of the
edges (defined by the weights that annotate them) traversed in this infinite
path, while Player 2 tries to minimize this average. It has been shown
in~\cite{em79,bsv04} that both players in an MPG can play optimally using
memoryless strategies, and as a consequence, those games are known to be
solvable in $\NP \cap \coNP$, and although pseudo-polynomial-time algorithms to
solve these games are known~\cite{zp96,bcdgr11}, the question of whether they
can be solved in polynomial-time is an important open question.

In the version of MPG described above, the game is of \emph{perfect information}:
both players have complete knowledge of the history of the play up to the
current position of the token. For many applications such as controller
synthesis, it is often more natural to assume that players have only partial
knowledge of the current state of the game. In practice,
players may model processes with private variables that other players
(processes) may not see, or controllers that acquire information about their
environment using sensors with bounded precision, etc. Unfortunately, it has
been shown in~\cite{ddgrt10} that MPGs games with \emph{partial-observation} are
undecidable.

\emph{Window mean-payoff objectives} (WMP objectives) were recently introduced
in~\cite{cdrr13} as an alternative to the classical MP objectives. In a
WMP objective instead of considering the long-run average along
the whole play, payoffs are considered over a local bounded window \emph{sliding}
along the play. The objective is then to make sure that the average payoff is at
least zero over every window. The WMP objectives enjoy several nice
properties. First, in contrast to classical MP games, we have a
polynomial-time algorithm for determining WMP games. Second, they
can be considered as approximation of the classical MP objectives in
the following sense: $(i)$ they are a \emph{strengthening} of the MP
objective, i.e. winning for the WMP objective implies winning for
the MP objective, $(ii)$ if a (finite memory) strategy is forcing a
MP with value $\varepsilon > 0$ then that strategy is also enforcing
the WMP objective for a window size that is bounded by a function
of the game and strategy memory sizes. 

In this paper we consider the extension of WMP objectives to
games with partial-observation.  We show that, in sharp contrast with classical
MP objectives, some of the WMP objectives are decidable for such
games. As in~\cite{cdrr13}, we consider several variants of the window
MP objectives.  For all objectives, we provide complete complexity results and optimal algorithms. More precisely, our main contributions are as follows:
\begin{itemize}
\item First, we consider a definition in which the window size is fixed and the
	sliding window is started at the initial move of the game, this is called the
	\emph{direct window} objective. For this definition we give an optimal
	\EXP-time algorithm (Theorem~\ref{thm:dirfix-exp-complete}) in the form
	of a reduction to a \emph{safety game}.  Additionally, we show that this
	safety game has a nice structure that induces a natural partial order on
	game positions. In turn this partial order can be used to obtain a
	{\em symbolic algorithm} based on the antichain approach~\cite{DoyenR10}.
	This shows that WMP objectives allow us not only to
	recover decidability but they also lead to games that have the potential
	to be solved efficiently in practice. The antichain approach has already
	been applied and implemented with success for LTL
	synthesis~\cite{bbfjr12}, omega-regular games with partial
	observation~\cite{BerwangerCWDH09}, and language inclusion between
	non-determinisitic B\"uchi automata~\cite{DR09}. 
\item Second, we consider two natural prefix-independent definitions for the
	window objectives, the \emph{(uniform) fixed window} objectives.  We
	also give optimal \EXP-time algorithms for these two definitions
	(Theorem~\ref{thm:fix-exp-complete} and
	Theorem~\ref{thm:ufix-exp-complete}), when weights are polynomially
	bounded in the size of the game arena. For these objectives, we show
	that the sets of good abstract plays (sequences of observations) of the
	game with partial observation, for the two definitions, form regular
	languages whose complements can be recognized by non-deterministic
	B\"uchi automata of pseudo-polynomial size
	(Proposition~\ref{pro:fix-lang} and Proposition~\ref{pro:ufix-lang}).
	These automata can then be turned into deterministic parity automata
	that can be used as observers to transform the game of
	partial-observation into a game of perfect information with a parity
	objective. 
\item Finally, we show that, when the size of the window is not fixed but rather
	left as a parameter, then for all the objectives that we consider the
	decision problems are undecidable
	(Theorem~\ref{thm:bwmp-undec}).
\end{itemize}


\section{Preliminaries}

\paragraph*{Weighted game arenas.}
A \emph{weighted game arena with partial-observation} (WGA, for brevity) is a
tuple $G = \langle Q, q_I, \Sigma, \Delta, w, Obs \rangle$, where $Q$ is a
finite set of states, $q_I \in Q$ is the initial state, $\Sigma$ is a finite set
of actions, $\Delta \subseteq Q \times \Sigma \times Q$ is the transition
relation, $w:\Delta \to \mathbb{Z}$ is the weight function, and $Obs \in
\partition(Q)$ is a set of observations containing $\{q_I\}$.
Let $W = \max\{ |w(t)| : t \in \Delta\}$. We assume $\Delta$ is total, i.e. for
every $(q,\sigma) \in Q \times \Sigma$ there exists $q' \in Q$ such that
$(q,\sigma,q') \in \Delta$. If every element of $Obs$ is a singleton, then we
say $G$ is a \emph{WGA with perfect information} and if $|Obs| = 1$ we say $G$
is \emph{blind}. For simplicity, we denote by $\post_\sigma(s) = \{q' \in Q \st
\exists q \in s : (q, \sigma, q') \in \Delta \}$ the set of $\sigma$-successors
of a set of states $s \subseteq Q$. 

In this work, unless explicitly stated otherwise, we depict states from a WGA
as circles and transitions as arrows labelled by an action-weight pair:
$\sigma,x \in \Sigma \times \{-W, \dots, W\}$. Observations are represented by
dashed boxes and colors, where states with the same color correspond to the same
observation.

\paragraph*{Abstract and concrete paths.} A \emph{concrete path} in an WGA is a
sequence $q_0 \sigma_0 q_1 \sigma_1 \dots$ where for all $i \ge 0$ we have $q_i
\in Q$, $\sigma_i \in \Sigma$ and $(q_i,\sigma_i,q_{i+1}) \in \Delta$. An
\emph{abstract path} is a sequence $o_0 \sigma_0 o_1 \sigma_1 \dots$ where $o_i
\in Obs$, $\sigma_i \in \Sigma$ and such that there is a concrete path $q_0
\sigma_0 q_1 \sigma_1 \dots$ for which $q_i \in o_i$, for all $i$.  Given an
abstract path $\psi$, let $\gamma(\psi)$ be the set of concrete paths that agree
with the observation and action sequence. Formally $\gamma(\psi) = \{ q_0
\sigma_0 q_1 \sigma_1 \dots \st \forall i \ge 0 : q_i \in o_i \text{
and } (q_i, \sigma, q_{i+1}) \in \Delta\}$.  Also, given abstract (respectively
concrete) path $\rho = o_0 \sigma_0 \dots$ and integers $k,l$ we define
$\pi[k..l] = o_k\dots o_l$, $\pi[..k] = \pi[0..k]$, and $\pi[l..] = o_l
\sigma_l o_{l+1}\dots$.

Given a concrete path $\pi = q_0\sigma_0 q_1 \sigma_1 \dots$, the \emph{payoff} 
up to the $(n+1)$-th element is given by
\[
	w(\pi[..n]) = \sum\limits_{i=0}^{n-1} w(q_i, \sigma_i, q_{i+1}).
\]
If $\pi$ is infinite, we define two \emph{mean-payoff} values $\mpinf$
and $\mpsup$ as:

\begin{minipage}{0.45\linewidth}
\begin{equation*}
\label{eqn:payoff}
	\mpinf(\pi) = \liminf\limits_{n \rightarrow \infty} \frac{1}{n}
	w(\pi[..n])
\end{equation*}
\end{minipage}
\hspace{0.2cm}
\begin{minipage}{0.45\linewidth}
\begin{equation*}
\label{eqn:mean-payoff}
	\mpsup(\pi) = \limsup\limits_{n \rightarrow \infty} \frac{1}{n}
	w(\pi[..n])
\end{equation*}
\end{minipage}

\paragraph*{Plays and strategies.}
A play in a WGA $G$ is an infinite abstract path
starting at $o_I \in Obs$ where $q_I \in o_I$.  Denote by $\plays(G)$ the set of
all plays and by $\prefs(G)$ the set of all finite prefixes of such plays ending
in an observation. Let $\gamma(\plays(G))$ be the set of concrete paths of all
plays in the game, and $\gamma(\prefs(G))$ be the set of all finite prefixes of
all concrete paths. 

An \emph{observation-based strategy for \eve} is a function from finite prefixes
of plays to actions, i.e. $\lambda_\exists : \prefs(G) \to \Sigma$. A play $\psi
= o_0 \sigma_0 o_1 \sigma_1 \dots$ is \emph{consistent} with $\lambda_\exists$
if $\sigma_i = \lambda_\exists(\psi[..i])$ for all $i$. We say an
observation-based strategy for \eve $\lambda_{\exists}$ has \emph{memory $\mu$} if
there is a set $M$ with $|M|=\mu$, an element $m_0 \in M$, and functions
$\alpha_u:M \times Obs \to M$ and $\alpha_o:M \times Obs \to \Sigma$ such that
for any play prefix $\rho = o_0 \sigma_0 \dots o_n$ we have
$\lambda_{\exists}(\rho) = \alpha_o(m_n,o_n)$, where $m_n$ is defined
inductively by $m_{i+1} = \alpha_u(m_i,o_i)$ for $i\geq 0$.

\paragraph*{Objectives.}

An \emph{objective} for a WGA $G$ is a set $V_G$ of plays, i.e. $V_G \subseteq
\plays(G)$. We say plays in $V_G$ are \emph{winning for \eve}. Conversely, all
plays not in $V_G$ are \emph{winning for \adam}. We refer to a WGA with a fixed
objective as a game. Having fixed a game, we say a strategy $\lambda$ is
\emph{winning} for a player if all plays consistent with $\lambda$ are winning
for that player. Finally, we say that a player \emph{wins} a game if (s)he has a
winning strategy. We write $V$ instead of $V_G$ if $G$ is clear from the
context.

Given WGA $G$ and a threshold $\nu \in \mathbb{Q}$, the \emph{mean-payoff} (MP)
objectives $\mpobjsup_G(\nu) = \{ \psi \in \plays(G) \st \forall \pi \in
\gamma(\psi) : \mpsup(\pi) \ge \nu \}$ and $\mpobjinf_G(\nu) = \{ \psi \in
\plays(G) \st \forall \pi \in \gamma(\psi) : \mpinf(\pi) \ge \nu \}$
require the mean-payoff value be at least $\nu$. We omit the subscript
in the objective names when the WGA is clear from the context. Let $\nu
= \frac{a}{b}$, $w'$ be a weight function mapping $t \in \Delta$ to $b
\cdot w(t) - a$, for all such $t$, and $G'$ be the WGA resulting from
replacing $w'$ in $G$ for $w$. We note that \eve wins the
$\mpobjsup_{G'}(0)$ (respectively, $\mpobjinf_{G'}(0)$) objective if and
only if she wins $\mpobjsup_G(\nu)$ (resp., $\mpobjinf_G(\nu)$).

\section{Window Mean-payoff Objectives}
In what follows we recall the definitions of the \emph{window mean-payoff} (WMP)
objectives introduced in~\cite{cdrr13} and adapt them to the partial-observation
setting. For the classical MP objectives \eve is required to ensure the long-run
average of all concretizations of the play is at least $\nu$. WMP objectives
correspond to conditions which are sufficient for this to be the case. All of
them use as a main ingredient the concept of concrete paths being ``good''. We
formalize this notion below.

Given $i \ge 0$ and \emph{window size bound} $\lmax \in \mathbb{N}_0 =
\mathbb{N} \setminus \{0\}$, let the set of concrete paths $\chi$ with 
the \emph{good window} property be
\begin{equation*}
	\gw(\nu, i, \lmax) = \{ \chi \st \exists j \le \lmax :
		\frac{w(\chi[i..(i + j)])}{j} \ge \nu \}.
\end{equation*}
As in~\cite{cdrr13}, we assume that the value of $\lmax$ is polynomially
bounded by the size of the arena.

For the first of the WMP objectives \eve is required to ensure that all
suffixes of all concretizations of the play can be split into concrete paths of
length at most $\lmax$ and average weight at least $\nu$. The MP objectives are
known to be prefix-independent, therefore a prefix-independent version of this first
objective is a natural objective to consider as well. We study two such
candidates. One which asks of \eve that there is some $i$ such that all
suffixes -- after $i$ -- of all concretizations of the play can be split in the
same way as before. This is quite restrictive since the $i$ is \emph{uniform}
for all concretizations of the play. The second prefix-independent version of
the objective we consider allows for non-uniformity.

Formally, the \emph{fixed window mean-payoff} (FWMP)
objectives for a given WGA and threshold $\nu \in \mathbb{Q}$ are defined
below. For convenience we denote by $\psi$ plays from $\plays(G)$ and concrete
plays by $\pi$, i.e. elements of $\gamma(\plays(G))$.
\begin{align*}
	\dirfix(\nu, \lmax) &= \{ \psi \st \forall \pi \in
		\gamma(\psi), \forall i \ge 0: \pi \in \gw(\nu, i, \lmax) \} \\
	\ufix(\nu, \lmax) &= \{ \psi \st \exists i \ge 0, \forall \pi \in 
		\gamma(\psi), \forall j \ge i : \pi \in \gw(\nu, j, \lmax) \} \\
	\fix(\nu, \lmax) &= \{ \psi \st \forall \pi \in \gamma(\psi), \exists i
		\ge 0, \forall j \ge i: \pi \in \gw(\nu, j, \lmax) \}
\end{align*}

For the FWMP objectives, we consider $\lmax$ to be a value that is given as
input. Another natural question that arises is whether we can remove this input
and consider an even weaker objective in which one asks if there exists an
$\lmax$. This is captured in the definition of the 
\emph{bounded window mean-payoff} (BWMP) objectives which are defined for a given
threshold $\nu \in \mathbb{Q}$.
\begin{align*}
	\udirbnd(\nu) &= \{ \psi \st \exists \lmax \in \mathbb{N}_0,\forall \pi
		\in \gamma(\psi), \forall i \ge 0: \pi \in \gw(\nu, i, \lmax) \}
		\\
	\dirbnd(\nu) &= \{ \psi \st \forall \pi \in \gamma(\psi), \exists \lmax
		\in \mathbb{N}_0, \forall i \ge 0: \pi \in \gw(\nu, i, \lmax) \}
		\\
	\ubnd(\nu) &= \{ \psi \st \exists \lmax \in \mathbb{N}_0, \exists i \ge 0,
\forall \pi \in \gamma(\psi), \forall j \ge i: \pi \in \gw(\nu, j, \lmax) \} \\
	\bnd(\nu) &= \{ \psi \st \forall \pi \in \gamma(\psi), \exists \lmax
		\in \mathbb{N}_0, \exists i \ge 0, \forall j \ge i: \pi \in
		\gw(\nu, j, \lmax) \}
\end{align*}

As with the mean-payoff objectives we can assume, without loss of generality,
that $\nu = 0$. Henceforth, we omit $\nu$.

\subsection{Relations among objectives}
Figure~\ref{fig:relations} gives an overview of the relative strengths of each
of the objectives and how they relate to the mean-payoff objective.  
The strictness, in general, of most inclusions was established in~\cite{cdrr13}, 
and Figure~\ref{fig:fix-neq-ufix} provides an example for the remaining case 
between $\fix$ and $\ufix$. 

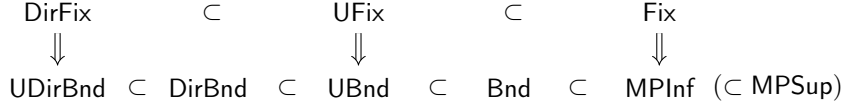
\begin{figure}
\begin{center}
\begin{tikzpicture}[x=2cm,y=1cm, implies/.style={double,double equal sign distance,-implies}]

	\node(A)at(0,1){$\dirfix$};
	\node(B)at(2,1){$\ufix$};
	\node(C)at(4,1){$\fix$};

	\node(D)at(0,0){$\udirbnd$};
	\node(E)at(1,0){$\dirbnd$};
	\node(F)at(2,0){$\ubnd$};
	\node(G)at(3,0){$\bnd$};
	\node(H)at(4,0){$\mpobjinf$};
	\node[anchor=west](I)at(H.east){($\subset \mpobjsup$)};

	\path
	(A) edge[draw=none, auto=false] node{$\subset$} (B)
	(B) edge[draw=none, auto=false] node{$\subset$} (C)
	(D) edge[draw=none, auto=false] node{$\subset$} (E)
	(E) edge[draw=none, auto=false] node{$\subset$} (F)
	(F) edge[draw=none, auto=false] node{$\subset$} (G)
	(G) edge[draw=none, auto=false] node{$\subset$} (H)
	(A) edge[implies] (D)
	(B) edge[implies] (F)
	(C) edge[implies] (H);
\end{tikzpicture}
\caption{Implications among the objectives}
\label{fig:relations}
\end{center}
\end{figure}

\begin{figure}
\begin{center}
\begin{minipage}{0.45\linewidth}
\begin{center}
\begin{tikzpicture}
	\node[state,initial,fill=blue!20](A){$q_0$};
	\node[state,fill=blue!20](B)[right=of A]{$q_1$};

	\node[fit=(A) (B)]{};

	\path
	(A) edge[loop] node[el,swap] {$\Sigma$,0} (A)
	(A) edge node[el] {$\Sigma$,-1} (B)
	(B) edge[loop] node[el,swap] {$\Sigma$,0} (B);
\end{tikzpicture}
\caption{Blind WGA where, for any $\lmax \in \mathbb{N}_{0}$, the only possible
abstract play is in $\fix(\lmax)$ but not in $\ufix(\lmax)$.}
\label{fig:fix-neq-ufix}
\end{center}
\end{minipage}
\hspace{0.2cm}
\begin{minipage}{0.45\linewidth}
\begin{center}
\begin{tikzpicture}
	\node[state,initial,fill=blue!20](A){$q_0$};
	\node[state,fill=green!20](B)[right=of A]{$q_1$};

	\node[fit=(A)]{};
	\node[fit=(B)]{};

	\path
	(A) edge[bend left] node[el] {$\Sigma$,-1} (B)
	(B) edge[bend left] node[el,swap] {$\Sigma$,1} (A)
	(B) edge[loop] node[el,swap] {$\Sigma$,0} (B);
\end{tikzpicture}
\caption{Perfect information WGA where \eve wins both MP objectives but
	none of the FWMP or BWMP objectives.}
\label{fig:mp-not-wmp}
\end{center}
\end{minipage}
\end{center}
\end{figure}

In general the mean-payoff objective is not sufficient for the FWMP or BWMP
objectives, e.g. see Figure~\ref{fig:mp-not-wmp}. Our first result shows that
if, however, \eve has a \emph{finite memory} winning strategy for a
\emph{strictly positive} threshold, then this strategy is also winning for any
of the FWMP or BWMP objectives. A specific subcase of this was first observed in
Lemma $2$ of~\cite{cdrr13}. 

\begin{theorem}
	\label{thm:mpfm-then-bnd}
	Given WGA $G$, if \eve has a finite memory winning strategy for the
	$\mpobjinf(\epsilon)$ (or $\mpobjsup(\epsilon)$)~objective, for
	$\epsilon > 0$, then the same strategy is winning for her in the
	$\dirfix(\mu)_G$~game -- where $\mu$ is bounded by the memory used by
	the strategy.
\end{theorem}

\subsection{Lower bounds}\label{sec:lowerBound}

In~\cite{cdrr13} it was shown that in multiple dimensions, with arbitrary window
size, solving games with the (direct) fixed window objective was complete for
\EXP-time.  We now show that in our more general setting this hardness result
holds, even when the window size is a fixed constant and the weight function is given in unary.

\begin{lemma}
	\label{lem:dirfix-low}
	Let $\lmax\in \mathbb{N}_0$ be a fixed constant.  Given WGA $G$, determining if \eve has a winning strategy for the
	$\dirfix(\lmax)$, $\ufix(\lmax)$ or the $\fix(\lmax)$~objectives
	is \EXP-hard, even for unary weights.
\end{lemma}

\begin{proof}
	We give a reduction from the problem of determining the winner of
	a safety game with imperfect information, shown in~\cite{cd10} to
	be \EXP-complete.

	A safety game with imperfect information is played on a non-weighted
	game arena with partial-observation $G = \langle Q, q_I, \Sigma, \Delta,
	Obs \rangle$. A play of $G$ is winning for \eve if and only if it never
	visits the \emph{unsafe state} set $\calU \subseteq Q$. Without
	loss of generality, we assume unsafe states are \emph{trapping}, that is
	$(u, \sigma, q) \in \Delta$ and $u \in \calU$ imply that $u = q$.

	Let $w$ be the transition weight function mapping $(u,\sigma,q) \in
	\Delta$ to $-1$ if $u \in \calU$ and all other $t \in \Delta$ to $0$.
	Denote by $G_w$ the resulting WGA from adding $w$ to $G$.  It should be
	clear that \eve wins the safety game $G$ if and only if she wins
	$\mpobjinf_{G_w}(0)$, $\dirfix_{G_w}(\lmax)$, $\ufix_{G_w}( \lmax)$,
	and $\fix_{G_w}( \lmax)$ -- for any $\lmax$. That is, all objectives
	are equivalent for $G_w$.
\qed\end{proof}

In~\cite{cdrr13} the authors show that determining if \eve has a winning
strategy in the $k$-dimensional version of the \udirbnd~and \ubnd~objectives
with perfect information is non-primitive recursive hard. We show that, in our
setting, these decision problems are undecidable.

\begin{theorem}
	\label{thm:bwmp-undec}
	Given WGA $G$, determining if \eve has a winning strategy for any of the
	BWMP objectives is undecidable, even if $G$ is blind.
\end{theorem}

\begin{figure}
\begin{center}
\begin{minipage}{0.45\linewidth}
\begin{center}
\begin{tikzpicture}
	\node[state,initial,fill=blue!20](A){$q_1$};
	\node[state,fill=blue!20](B)[right=of A]{$q_2$};
	\node[state,fill=blue!20](C)[right=of B]{$q_3$};

	\node[fit=(A) (B) (C)]{};

	\path
	(A) edge[loop] node[el,swap] {$\Sigma\cup\{\#\}$,0} (A)
	(A) edge node[el] {$\Sigma$,-1} (B)
	(B) edge[loop] node[el,swap] {$\Sigma$,-1} (B)
	(B) edge node[el] {$\#$,1} (C)
	(C) edge[loop] node[el,swap] {$\Sigma$,1} (C)
	;
\end{tikzpicture}
\caption{Gadget which forces \eve to play infinitely many $\#$.}
\label{fig:infty-hash}
\end{center}
\end{minipage}
\hspace{0.2cm}
\begin{minipage}{0.45\linewidth}
\begin{center}
\begin{tikzpicture}
	\node[state,initial,fill=blue!20](A){$q_4$};
	\node[state,fill=blue!20](B)[right=of A]{$q_5$};

	\node[fit=(A) (B)]{};

	\path
	(A) edge[loop] node[el,swap] {$\Sigma\cup\{\#\}$,0} (A)
	(A) edge[bend left] node[el] {$\#$,-1} (B)
	(B) edge[loop] node[el,swap] {$\Sigma$,0} (B)
	(B) edge[bend left] node[el,swap] {$\#$,1} (A)
	;
\end{tikzpicture}
\caption{Gadget which, given that \eve will play $\#$ infinitely often, forces
her to play $\#$ in intervals of bounded length.}
\label{fig:no-asc-chains}
\end{center}
\end{minipage}
\end{center}
\end{figure}

\begin{figure}
\begin{center}
\begin{tikzpicture}
	\node[state,initial,fill=blue!20](A){$q_5$};
	\node[state,fill=blue!20](B)[right=of A]{$q_I$};
	\node[state,fill=blue!20](C)[right=of B]{$p \not\in F$};
	\node[state,fill=blue!20](D)[below=of B]{$q \in F$};
	\node[state,fill=blue!20](E)[right=of C]{$\bot$};

	\node[draw,fit=(B) (C) (D),label=right:$\mathcal{N}$,dashed](F){};
	\node[draw,dotted,fit=(A) (E) (F)]{};

	\path
	(A) edge node[el] {$\#$,$\frac{1}{2}$} (B)
	(C) edge node[el] {$\#$,0} (E)
	(E) edge[loop] node[el,swap] {$\Sigma\cup\{\#\}$,1} (E)
	(D) edge[out=150,in=200] node[el] {$\#$,$\frac{1}{2}$} (B)
	;

\end{tikzpicture}
\caption{Blind gadget to simulate the weighted automaton $\mathcal{N}$.}
\label{fig:sim-automaton}
\end{center}
\end{figure}
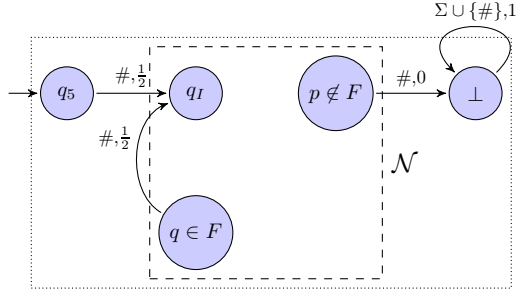

\begin{proof}[Sketch]
	Our proof is by a reduction from the universality problem of weighted
	finite automata which is undecidable~\cite{abk11}.  A \emph{weighted
	finite automaton} is a tuple $\mathcal{N}=\langle Q, \Sigma, q_I,
	\Delta, w, F \rangle$ where $F \subseteq Q$ is a set of final accepting
	states. A {\em (accepting) run} of the automaton on a word $x = \sigma_0
	\sigma_1 \dots \sigma_n \in \Sigma^*$ is a sequence $r = q_0 q_1 \dots
	q_n \in Q^+$ such that $q_n \in F$ and $(q_i, \sigma_i, q_{i+1}) \in
	\Delta$ for all $0 \le i < n$. The cost of the run $r$ is $w(r) =
	\sum^{n-1}_{i=0} w(q_i,\sigma_i,q_{i+1})$. If the automaton is
	non-deterministic, it may have several runs on $x$. In that case, the
	\emph{cost of $x$ in $\mathcal{N}$} (denoted by $\mathcal{N}(x)$) is
	defined as the minimum of the costs of all its accepting runs on $x$.

	The \emph{universality} problem for weighted automata is to decide
	whether, for a given automaton $\mathcal{N}$, the following holds:
	$\forall x \in \Sigma^* : \mathcal{N}(x) < 0$.	If so, $\mathcal{N}$ is
	{\em universal}.  Our reduction constructs a blind WGA, $G_\mathcal{N}$,
	so that:
	\begin{itemize}
		\item if $\mathcal{N}$ is universal, then \eve has a winning
			strategy for the objective \udirbnd,
		\item if $\mathcal{N}$ is not universal, then \adam has a
			winning strategy for the complement of the objective
			\bnd. 
	\end{itemize}
	As shown in Fig.~\ref{fig:relations}, $\udirbnd \subseteq \bnd$ and all
	the other BWMP objectives lie in between those two. So, our reduction
	establishes the undecidability of all BWMP objectives at once.
	
	Our reduction follows the gadgets given in
	Figs.~\ref{fig:infty-hash}--\ref{fig:sim-automaton}.
	When the game starts, \adam chooses to play from one of the three
	gadgets. As the game is blind for \eve, she does not know what is the
	choice of \adam and so she must be prepared for all possibilities. Note
	also that as \eve is blind, her strategy can be formalized by an
	infinite word $w \in \Sigma \cup \{ \# \}^{\omega}$.  The first two
	gadgets force \eve to play a word $w$ such that: $(C_1)$ there are
	infinitely many $\#$ in $w$, and $(C_2)$ there exists a bound $b \in
	\mathbb{N}$ such that the distance between two consecutive $\#$ in $w$
	is bounded by $b$.  Note that although $(C_1)$ holds with under a 
	mean-payoff objective, the \bnd\ objective is necessary to ensure $(C_2)$.
	
	Now, we assume that \eve plays a word $w=\#w_1 \# w_2\# w_3\# \dots \#
	w_n \# \dots$ that respects conditions $C_1$ and $C_2$, and we consider
	what happens when \adam plays in the third gadget
	(Fig.~\ref{fig:sim-automaton}). The best strategy for \adam is to
	simulate the accepting runs of $\mathcal{N}$. If $\mathcal{N}$ is not
	universal then \eve chooses $w_1 \in \Sigma^*$ such that
	$\mathcal{N}(w_1) \geq 0$, and then she plays the strategy
	$(\#w_1)^{\omega}$ and she wins the objective \udirbnd. Indeed \eve
	plays a finite memory strategy that forces a mean-payoff which is at
	least $\frac{0.5}{b} > 0$ as each new $\#$ brings $+\frac{1}{2}$ and we
	know that $\mathcal{N}(w_1)\geq0$. The result then follows from
	Theorem~\ref{thm:mpfm-then-bnd}.
	If $\mathcal{N}$ is universal, no matter which word $w$ is played by
	\eve, the mean-payoff of the outcome of the game will be at most
	$\frac{-0.5}{b}$ as any word $w_1 \in \Sigma^*$ played between
	two consecutive $\#$ is such that $\mathcal{N}(w_1)\leq-1$. Thus, \adam
	wins for the complement of the mean-payoff objective. From the
	relation between the mean-payoff objective and \bnd\/ (see Fig.
	~\ref{fig:relations}) it follows that \adam wins for the complement of
	the objective \bnd. 
\qed\end{proof}

\section{Solving \dirfix~games}
\label{sec:dirfix-games}
In this section we establish an upper bound to match our lower bound of
Section~\ref{sec:lowerBound} for determining the winner of \dirfix~games.  We
first observe that for WGAs with perfect information the $\dirfix(\lmax)$
objective has the flavor of a safety objective. Intuitively, a play $\pi$ is
winning for \eve if every suffix of $\pi$ has a prefix of length at most $\lmax$
with average weight of at least $0$. As soon as the play reaches a point for
which this does not hold, \eve loses the play. In WGAs with partial-observation
we need to make sure the former holds for all concretizations of an abstract
play.

We construct a non-weighted game arena with perfect information $G'$ from $G$.
\eve's objective in $G'$ will consist in ensuring the play never reaches
locations in which there is an open window of length $\lmax$, for some state.
This corresponds to a safety objective. Whether \eve wins the new
game can be determined in time linear w.r.t. the size of the new game (see,
e.g.~\cite{thomas95}). 
The game will be played on a set of functions $\mathcal{F}$ which is described
in detail below. We then show how to transfer winning strategies of \eve from
$G'$ to $G$ and vice versa in Lemmas~\ref{lem:dirfix-cor}
and~\ref{lem:dirfix-comp}. Hence, this yields an algorithm to determine if \eve
wins the $\dirfix(\lmax)$ objective which runs in exponential time.

\begin{theorem}
	\label{thm:dirfix-exp-complete}
	Given WGA $G$, determining if \eve has a winning strategy for the
	$\dirfix(\lmax)$~objective is \EXP-complete.
\end{theorem}

Let us define the functions which will be used as the state space of the game.
Intuitively, we keep track of the \emph{belief} of \eve as well as the windows
with the minimal weight open at every state of the belief.\footnote{The terms
belief and knowledge are used to denote a state from any variation of
the classic ``Reif construction''~\cite{reif84} to turn a game with
partial-observation into a game with perfect information.  Other names
for similar constructions include ``knowledge-based subset
construction'' (see e.g.~\cite{ddgrt10}).}

For the rest of this section let us fix a WGA with partial-observation $G$ and a
window size bound $\lmax \in \mathbb{N}_0$. We begin by defining the set of
functions $\mathcal{F}$ as the set of all functions $f : Q \to ( \{1, \dots,
\lmax\} \to \{-W \cdot \lmax, \dots, 0\} ) \cup \{\bot\}$. Denote by $\supp(f)$
the \emph{support} of $f$, that is the set of states $q \in Q$ such that $f(q)
\neq \bot$. For $q \in \supp(f)$, we denote by $f(q)_i$ the value $f(q)(i)$.
The function $f_I \in \mathcal{F}$ is such that $f_I(q_I)_l = 0$, for all $1 \le
l \le \lmax$, and $f_I(q) = \bot$ for all $q \in Q \setminus \{q_I\}$.  Given
$f_1 \in \mathcal{F}$ and $\sigma \in \Sigma$, we say $f_2 \in \mathcal{F}$ is
a $\sigma$-successor of $f_1$ if
\begin{itemize}
	\item $\supp(f_2) = \post_\sigma(\supp(f_1)) \cap o$ for some $o \in
		Obs$;
	\item for all $q \in \supp(f_2)$ and all $1 \le j \le \lmax$ we have
		that $f_2(q)_j$ maps to $\max\{-W \cdot \lmax, \min\{0,
		\zeta(q)\}\}$, where $\zeta(q)$ is defined as follows
		\[
			\zeta(q) = 
			\begin{cases}
				\min\limits_{\substack{p \in \supp(f_1),\\
		 	 		 	       (p,\sigma,q) \in \Delta,\\
						       f_1(p)_{j-1} < 0}}
				f_1(p)_{j-1} + w(p,\sigma,q) & \mbox{if } j \ge 2 \\
				\min\limits_{\substack{p \in \supp(f_1),\\
		 	 		 	       (p,\sigma,q) \in \Delta}}
				w(p,\sigma,q) & \mbox{otherwise.}
			\end{cases}
		\]
\end{itemize}

\begin{restatable}{lemma}{sizeofF}
\label{lem:sizeof-F}
	The number of elements in $\mathcal{F}$ is not greater than $2^{|Q|
	\cdot \lmax \cdot \log(W \cdot \lmax)}$.
\end{restatable}

We extend the $\supp$ operator to finite sequences of functions and actions. In
other words, given $\rho' = f_0 \sigma_0 f_1 \sigma_1 \in (\mathcal{F} \cdot
\Sigma)^*$, $\supp(\rho') \mapsto s_0 \sigma_0 s_1 \sigma_1 \dots$ where $s_i =
\supp(f_i)$ for all $i \ge 0$. In an abuse of notation, we define the function
$\supp^{-1} : (Obs \cdot \Sigma)^* \times \mathcal{F} \to (\mathcal{F} \cdot
\Sigma)^*$ which maps abstract paths to function-action sequences.  Formally,
given $\rho = o_0 \sigma_0 o_1 \sigma_1 \dots \in \prefs(G)$ and $\phi \in
\mathcal{F}$ with $\supp(\phi) \subseteq o_0$, $\supp^{-1}(\rho, \phi) \mapsto
f_0 \sigma_0 f_1 \sigma_1 \dots$ where $f_0 = \phi$ and for all $i \ge 0$ we
have that $f_{i+1}$ is the $\sigma_i$-successor of $f_i$ such that
$\supp(f_{i+1}) \subseteq o_{i+1}$. Both $\supp$ and $\supp^{-1}$ are extended
to infinite sequences in the obvious manner.

The following two results enunciate the key properties of sequences of the form
$(\mathcal{F} \cdot \Sigma)^*$. Intuitively, the set of all those sequences
corresponds to the windowed, weighted unfolding of $G$ with information about
reachable states as well as open windows.

\begin{restatable}{lemma}{suppisreachable}
\label{lem:supp-is-reachable}
	Let $\rho = o_0 \sigma_0 \dots o_n$ be an abstract path, $\phi \in
	\mathcal{F}$ such that $\supp(\phi) \subseteq o_0$ and $\supp^{-1}(\rho,
	\phi) =	f_0 \sigma_0 \dots f_n \in (\mathcal{F} \cdot \Sigma)^*$. A state $q
	\in Q$ is reachable from some state $q_0 \in \supp(\phi)$ through a concrete
	path $q_0 \sigma_0 \dots q_n \in \gamma(\rho)$ if and only if $q \in
	\supp(f_n)$.
\end{restatable}

Consider an abstract path $\psi$ and a positive integer $n$. We say a window of
length $l$ is open at $q \in \gamma(\psi[n])$ if there is some concretization
$\chi$ of $\psi[..n]$ with $q = \chi[n]$ such that $\chi \not\in \gw(n-l,l)$.

\begin{restatable}{lemma}{thekey}
\label{lem:the-key}
	Let $\rho = o_0 \sigma_0 \dots o_n$ be an abstract path, $\phi \in
	\mathcal{F}$ such that $\supp(\phi) \subseteq o_0$ and $\supp^{-1}(\rho,
	\phi) = f_0 \sigma_0 \dots f_n \in (\mathcal{F} \cdot \Sigma)^*$.
	Given state $p \in \supp(f_n)$ and $1 \le l \le \lmax$ such that $l \le
	n$, then there is a window of length $l$ open at $p$ if and only if
	$f_n(p)_l < 0$.
\end{restatable}

Formally, the arena $G' = \langle \mathcal{F}, f_I, \Sigma, \Delta'
\rangle$. The transition relation $\Delta'$ contains the transition $(f_1,
\sigma, f_2)$ if $f_2$ is the $\sigma$-successor of $f_1$.  \eve, in $G'$, is
required to avoid states $\calU = \{f \in \mathcal{F} \st \exists q \in \supp(f) :
f(q)_\lmax < 0 \}$. 

\begin{lemma}
	\label{lem:dirfix-cor}
	If \eve wins the safety objective in $G'$, then she also wins the
	$\dirfix(\lmax)$~objective in $G$.
\end{lemma}

\begin{proof}
	Assume $\lambda'$ is a winning strategy for \eve in $G'$. We define a
	strategy $\lambda$ for her in $G$ as follows:
	\[
		\lambda( \rho ) \mapsto \lambda'( \supp^{-1}( \rho, f_I))
	\]
	for all $\rho \in \prefs(G)$. We claim that $\lambda$ is winning for her
	in $G$. Towards a contradiction, assume $\psi \in \plays(G)$ is
	consistent with $\lambda$ and that $\psi \not\in \dirfix(\lmax)$.
	Recall that this implies there are $n \in \mathbb{N}, q \in Q$ such that
	there is a window of length $\lmax$ open at $q \in \gamma(\psi[n])$.
	By Lemma~\ref{lem:the-key} we then get that $f_n$ from $\supp^{-1}(\psi,
	f_I) = f_0 \sigma_0 f_1 \sigma_1 \dots$ is in $\calU$.  As
	$\supp^{-1}(\psi, f_I)$ is consistent with $\lambda'$, this contradicts
	the assumption that $\lambda'$ was winning.
\qed\end{proof}

\begin{lemma}
	\label{lem:dirfix-comp}
	If \eve wins the $\dirfix(\lmax)$~objective in $G$, then she also wins the
	safety objective in $G'$.
\end{lemma}

\begin{proof}
	Assume $\lambda$ is a winning strategy for \eve in $G$. We define a
	strategy $\lambda'$ for her in $G'$ as follows:
	\[
		\lambda'( \rho' ) \mapsto \lambda \circ \obs \circ \supp( \rho')
	\]
	for all $\rho' \in \prefs(G')$. We claim that $\lambda'$ is winning for
	her in $G'$. Again, towards a contradiction, assume $\psi' \in
	\plays(G')$ is consistent with $\lambda'$ and that $\psi'$ visits some
	$f \in \calU$. This implies, by Lemma~\ref{lem:the-key}, that there is a
	window of length $\lmax$ open at some $q \in \supp(f)$ in $\psi =
	\obs(\supp(\psi'))$. As $\Delta$ is total, for any $\sigma \in \Sigma$
	\eve plays then there is valid $\sigma$-successor of $q$ that \adam can
	choose as the next state. Hence there is some $\chi \in \plays(G)$
	consistent with $\lambda$ such that $\chi$ and $\psi$ have the same
	prefix up $i_q$, where $q \in \gamma(\chi[i_q])$, and there is a
	concretization $\pi$ of $\chi$ such that $\pi[i_q] = q$. As $\chi$ is
	consistent with $\lambda$ and $\chi \not\in \dirfix(\lmax)$, this
	contradicts the fact that it was a winning strategy.
\qed\end{proof}

\subsection{A symbolic algorithm for \dirfix~games}
We note that state space of the construction $G'$ presented in
Section~\ref{sec:dirfix-games} admits an order such that if a state is smaller
than another state, according to said order, and \eve has a strategy to win from
the latter, then she has a strategy to win from the former. In this section we
formalize this notion by defining the order and, in line
with~\cite{cdhr06,bbfjr12}, propose an \emph{antichain}-based algorithm to solve
the safety game on $G'$.

We define the \emph{uncontrollable predecessors} operator $\upre :
\pow(\mathcal{F}) \to \pow(\mathcal{F})$ as follows:
\[
	\upre(S) = \{ p' \in \mathcal{F} \st \forall \sigma \in \Sigma, \exists
		q' \in S : (p', \sigma, q') \in \Delta' \}.
\]
For $S \in \pow(\mathcal{F})$, we denote by $\mu X. (S \cup \upre(X))$, the
\emph{least fixpoint} of the function $F : X \to S \cup \upre(X)$ in the
$\mu$-calculus notation (see~\cite{ej91}). Note that $F$ is defined on the
powerset lattice, which is finite. The following is a well-known result about
the relationship between safety games and the $\upre$ operator (see
e.g.~\cite{gradel04}).
\begin{proposition}
\label{pro:win-lose-safe}
	\eve wins a safety game with unsafe state set $\calU$ if and only if
	the initial state of the game is not contained in $\mu X. (\calU \cup
	\upre(X))$.
\end{proposition}

\begin{definition}[The partial order]
Given $f',g' \in \mathcal{F}$ we say $f' \preceq g'$ if and only if $\supp(f')
\subseteq \supp(g')$ and
\[
	\forall q \in \supp(f'), \forall i \in \{1,\dots,\lmax\}, \exists j \in
	\{i,\dots,\lmax\} : f'(q)_i \ge g'(q)_j.
\]
\end{definition}

An \emph{antichain} is a non-empty set $S \in \pow(\mathcal{F})$ such that for
all $x,y \in S$ we have $x \not\preceq y$. We denote by $\mathfrak{A}$ the set
of all antichains. Given $a,b \in \mathfrak{A}$, denote by $a \sqsubseteq b$ the
fact that $\forall x \in b, \exists y \in a : y \preceq x$. For $S \in
\pow(\mathcal{F})$ we denote by $\lfloor S \rfloor$ the set of minimal elements
of $S$, that is $\lfloor S \rfloor = \{ x \in S \st \forall y \in S : y \preceq
x$ implies $y = x\}$. Clearly $\lfloor S \rfloor$ is an antichain.

Given $S \in \pow(\mathcal{F})$ we denote by $S \upclose$ the
\emph{upward-closure} of $S$, that is $S \upclose = \{ t \in \mathcal{F} \st S
\preceq t \}$. We say a set $s \in \pow(\mathcal{F})$ is \emph{upward-closed} if
$S = S \upclose$. Note that $\lfloor S \rfloor \upclose = S \upclose$ and
therefore, if $S$ is upward-closed, the antichain $\lfloor S \rfloor$ is a
succinct representation of $S$.

\begin{lemma}
	\label{lem:uc-sets}
	The following assertions hold.
	\begin{enumerate}
		\item $\calU$ is upward-closed.
		\item If $S,T \in \pow(\mathcal{F})$ are two upward-closed sets,
			then $S \cup T$ is also upward-closed.
	\end{enumerate}
\end{lemma}

The usual way of showing an antichain algorithm works dictates that we now prove
the $\upre$ operator, when applied to upward-closed sets, outputs an
upward-closed set as well. Unfortunately, this is not true in our case. The
following example illustrates this difficulty.

\begin{example}
Consider the WGA from Figure~\ref{fig:mp-not-wmp} and let $\lmax =
2$. We note that the function $f$ such that $f(q_0) = \bot$ and $f(q_1)_1 = 1,
f(q_1)_2 = 0$ is in $\upre(\calU)$.  We also have that for the function $g$ such
that $g(q_0) = \bot$ and $g(q_1)_1 = 0, g(q_1)_2 = 1$ we get that $f \preceq g$.
It is easy to verify that $g \not \in \upre(\calU)$. Hence, $\upre(\calU)$ is
not upward-closed.
\end{example}

However, we claim 
that one can circumvent this issue by ignoring elements from $\calU$.  Thus we are able to prove that,
under some conditions, $\upre$ does preserve ``upward-closedness''.
\begin{restatable}{lemma}{upreisuc}
	\label{lem:upre-is-uc} Given upward-closed set $S \in \pow(\mathcal{F})$
	and $f,g \in \mathcal{F} \setminus \calU$, if $f \in \upre(S)$ and $f
	\preceq g$, then $g \in \upre(S)$.
\end{restatable}

We define a version of the uncontrollable predecessors' operator which
manipulates antichains instead of subsets of $\mathcal{F}$.  
\[
	\lfloor \upre \rfloor (a) = \lfloor \{ p' \in \mathcal{F} \setminus \calU
	\st \forall \sigma \in \Sigma, \exists q' \in a, \exists r' \in \mathcal{F}
	: (p', \sigma, r') \in \Delta' \land q' \preceq r' \} \rfloor
\]

Given $a,b \in \mathfrak{A}$ we denote by $a \sqcup b$ the \emph{least upper
bound} of $a$ and $b$, i.e. $a \sqcup b = \lfloor \{q' \in \mathcal{F} \st q' \in a
\text{ or } q' \in b \} \rfloor$. It is easy to check that $(a \sqcup b)
\upclose = a \upclose \cup b \upclose$ for any $a, b \in \mathfrak{A}$.

\begin{restatable}{theorem}{acalgodirfix}
\label{thm:ac-dirfix}
	Given WGA $G$, \eve wins the $\dirfix(\lmax)$~objective if and only
	if $\{q_I'\} \not\sqsupseteq \mu X. ( \lfloor \calU \rfloor \sqcup
	\lfloor \upre \rfloor(X))$.
\end{restatable}



\section{Solving \fix~games}
Since \fix~games are a prefix-independent version of \dirfix~games,
it seems logical to consider an analogue of the perfect information game from
the previous section with a prefix-independent condition.  Indeed, the reader
might be tempted to extend the approach used to solve \dirfix~games by replacing
the safety objective with a \emph{co-B\"{u}chi} objective in order to solve
\ufix~or \fix~games. However, we observe that although \eve winning in the resulting game
is sufficient for her to win the original \fix~game, it is not necessary. Indeed,
an abstract play visits states from $\calU$ infinitely often if and only if for
infinitely many $i$ there is a concretization of the play prefix up to $i$ which
violates $\gw(i,\lmax)$. Nevertheless, this does not imply there exists one
(infinite) concretization of the play which violates $\gw(i,\lmax)$ for
infinitely many $i$. Figure~\ref{fig:fix-is-dif} illustrates this
phenomenon.

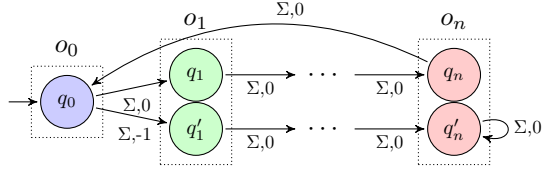
\begin{figure}
\begin{center}
\begin{tikzpicture}
	\node[state,initial,fill=blue!20](A){$q_0$};
	\node[state,fill=green!20](B)[right=of A, yshift=0.45cm]{$q_1$};
	\node[state,fill=green!20](C)[right=of A, yshift=-0.45cm]{$q_1'$};
	\node (D)[right=of B]{$\cdots$};
	\node (E)[right=of C]{$\cdots$};
	\node[state,fill=red!20](F)[right=of D]{$q_n$};
	\node[state,fill=red!20](G)[right=of E]{$q_n'$};

	\node[fit=(A),label=above:$o_0$]{};
	\node[fit=(B) (C),label=above:$o_1$]{};
	\node[fit=(F) (G),label=above:$o_n$]{};

	\path
	(A) edge node[el,swap,pos=0.3] {$\Sigma$,0} (B)
	(A) edge node[el,swap,pos=0.85] {$\Sigma$,-1} (C)
	(B) edge node[el,swap] {$\Sigma$,0} (D)
	(C) edge node[el,swap] {$\Sigma$,0} (E)
	(D) edge node[el,swap] {$\Sigma$,0} (F)
	(E) edge node[el,swap] {$\Sigma$,0} (G)
	(F) edge[bend right] node[el,swap,pos=0.4] {$\Sigma$,0} (A)
	(G) edge[loop right] node[el,swap] {$\Sigma$,0} (G);
\end{tikzpicture}
\caption{For $n>\lmax+1$ the abstract path $(o_0 \ldots o_n)^\omega$ is 
winning for the \fix~condition but infinitely often visits an unsafe 
state in the construction from Section~\ref{sec:dirfix-games}.}
\label{fig:fix-is-dif}
\end{center}
\end{figure}

For the reasons stated above, we propose to solve \fix~games in a different way.
We first introduce the notion of \emph{observer}. Let $\mathcal{A}$
be a deterministic parity automaton.\footnote{We refer the reader who is not
familiar with parity automata or games to~\cite{thomas95}.} We say
$\mathcal{A}$ is an \emph{observer for the objective} $V$ if the
language of $\mathcal{A}$ is $V$, i.e.  $\mathcal{L}(\mathcal{A}) = V$.
In~\cite{cd10}, the authors show that the synchronized product of $G$
and an observer for $V$ is a parity game of perfect information which is won by
\eve if and only if she wins $G$. Thus, it suffices to find an algorithm
to construct an observer for $\fix(\lmax)$ to be able to solve \fix~games.

For convenience, we start by describing a non-deterministic machine that accepts
as its language the complement of $\fix(\lmax)$. Note that all elements of
$\fix(\lmax)$ start with the observation $\{q_I\}$ so it suffices to describe
the machine that accepts any word $w \in (\Sigma \cdot Obs)^\omega$ such
that $\{q_I\} \cdot w \in \plays(G) \setminus \fix(\lmax)$. The machine we
construct is similar to the one used in~\cite{cd10} to make objectives of
imperfect information games visible.  Intuitively, at each step of the game and
after \adam has revealed the next observation we will guess his actual choice of
state using non-determinism.  Additionally, we shall guess whether or not a
violating window starts at the next step. The state space of the automaton will
therefore consist of a single state from $Q$, a negative integer to record the
weight of the tracked window, and the length of the current open window.

Formally, let $\mathcal{N}$ be the automaton consisting of the state space $F =
Q \times \{1,\dots,\lmax\} \times \{-W \cdot \lmax, \dots, -1 \} \cup
\{\bot\}$; initial state $(q_I,1,\bot)$; input alphabet $\Sigma' = \Sigma \times
Obs$; and $\Delta'' \subseteq F \times \Sigma' \times F$. The transition
relation $\Delta''$ has a transition $((p,i,n), (\sigma,o), (q,j,m))$ if $(p,
\sigma, q) \in \Delta$, $q \in o$,
\[
	m =
	\begin{cases}
		w(p,\sigma,q) & \mbox{if } w(p,\sigma,q) < 0 \\
		n + w(p,\sigma,q) & \mbox{if } n \neq \bot \land n +
			w(p,\sigma,q) < 0 \land i < \lmax \\
		\bot & \mbox{otherwise,}
	\end{cases}
\]
and
\[
	j = 
	\begin{cases}
		i + 1 & \mbox{if } m = n + w(p,\sigma,q) \\
		1 & \mbox{otherwise.}
	\end{cases}
\]
We say a state $(q,i,n) \in F$ is accepting if $i = \lmax$ and $n \neq \bot$.
The automaton accepts a word $x$ if and only it has a run $(q_0,i_0,n_0)$
$(\sigma_0,o_1)$ $(q_1,i_1,n_1)$ $(\sigma_1,o_2) \dots$ on $x$ such that for
infinitely many $j$ we have that $(q_j,i_j,n_j)$ is accepting. 

\begin{proposition}
	\label{pro:fix-lang}
	The non-deterministic B\"{u}chi automaton $\mathcal{N}$ accepts a word
	$\psi \in \plays(G)$ if and only if $\psi \not\in \fix(\lmax)$.
\end{proposition}

At this point we determinize $\mathcal{N}$ and complement it to get a
deterministic automaton with state space of size exponential in the size of
$\mathcal{N}$ yet has parity index that is polynomial w.r.t.\ the size of $Q$
(see~\cite{safra88,safra92,piterman07}). The synchronized product of $G$ and the
observer yields a parity game with the same size bounds. The desired result
follows from the parity games' algorithm and results of~\cite{jurdzinski00}.

\begin{theorem}
	\label{thm:fix-exp-complete}
	Given WGA $G$, determining if \eve has a winning strategy for the
	$\fix(\lmax)$~objective can be decided in time exponential in $W$ and the size of $G$.
\end{theorem}

\begin{corollary}
	Given WGA $G$ with unary encoded weights, deciding if \eve has a
	winning strategy for the $\fix(\lmax)$~objective is \EXP-complete.
\end{corollary}

\section{Solving \ufix~games}
In order to determine the winner of \ufix~games, we proceed as in the previous
section by finding a non-deterministic B\"uchi automaton that recognizes the set
of bad abstract plays.  However, in this case the situation is more complicated
because a bad abstract play might arise from a violation in the
\emph{uniformity}, rather than because of a concrete path with infinitely many
window violations.  Figure~\ref{fig:fix-neq-ufix} illustrates this issue.  To
overcome this, we first provide an alternative characterization of the bad
abstract plays for \eve.  Consider some $\psi \in \plays(G)$. We say $\pi \in
\gamma(\psi)$ \emph{merges with infinitely many violating paths} if for all $i
\ge 0$, there are $j \ge i$, $k \ge j + \lmax$ and some $\chi \in
\gamma(\psi[..k])$ such that $\pi[k] = \chi[k]$ and $\chi \not\in \gw(j,
\lmax)$.  We refer to $j$ as the position of the violation and to $k$ as the
position of the merge.  Our next result formally states the relationship between
concrete plays merging for multiple violations and \ufix~games.

\begin{lemma}
\label{lem:ufix-eq-infopenwin}
	Given WGA $G$ and $\psi \in \plays(G)$, there is $\pi \in \gamma(\psi)$
	merging with infinitely many violating paths if and only if $\psi \not\in
	\ufix(\lmax)$.
\end{lemma}
\begin{proof}
	($\Rightarrow$) Assume there is a $\pi \in \gamma(\psi)$ merging with
	infinitely many violating paths. We have that there are two infinite
	sequences of indices $J = \{j_0, j_1, \dots\}$ and $K = \{k_0, k_1,
	\dots\}$ such that $j_l < j_{l+1}$ and $j_l + \lmax \le k_l$, for all
	$l \ge 0$, and for which we know that there is 
	concrete path $\chi_l \in \gamma(\psi[..k_l])$ such that $\chi[j_l..j_l
	+ \lmax]$ realizes an open window of length $\lmax$ and $\chi_l[k_l] =
	\pi[k_l]$, for all $l \ge 0$. 
	Observe that for all $l \ge 0$ we have that 
	$\chi_l \cdot \pi[k_l..] \in \gamma(\psi)$ and that
	$\chi_l \cdot \pi[k_l..] \not\in \gw(j_l,\lmax)$. In other words,
	\[
		\forall l \ge 0, \exists \alpha \in \gamma(\psi), \exists m \ge
		l : \alpha \not\in \gw(m, \lmax)
	\]
	which implies that $\psi \not\in \ufix(\lmax)$.

	($\Leftarrow$) Assume $\psi \not\in \ufix(\lmax)$. We have that there is
	a infinite sequence of indices $J = \{j_0, j_1, \dots\}$ such that $j_k
	< j_{k+1}$, for all $k \ge 0$, and for which we know there is a concrete
	play $\pi_k \in \gamma(\psi)$ such that $\pi_k \not\in \gw(j_k - \lmax,
	\lmax)$, for all $k \ge 0$. Observe that for all $i \ge 0$ the set
	$\gamma(\psi[i])$ is finite and bounded by $|Q|$. Thus, by Pigeonhole
	Principle we have that, for all $n \ge 0$ there is $\eta_n \in \{\pi_m
	\st 1 \le m \le |Q| \cdot n \} \subseteq \gamma(\psi)$ which merges
	with at least $n$ violating paths. Consider an arbitrary $\eta_1$. If
	$\eta_1$ merges with infinitely many violating paths then we are done
	and the claim holds.  Otherwise it only merges with a finite number of
	violating paths, say $a_1$. From the previous argument we know there is
	an $\eta_{a'_1} \in \gamma(\psi)$ that merges with at least $a'_1 =
	a_1 + 1$. Clearly $\eta_1$ and $\eta_{a'_1}$ are disjoint at every point
	after $a'_1$, lest $\eta_1$ would merge with a new violating path. We
	inductively repeat the process, if $\eta_{a'_i}$ merges with infinitely
	many violating paths then we are done. Otherwise it only merges with
	some finite number of violating paths, say $a_i$. In that case we turn
	our attention to $\eta_{a'_{i+1}}$.  Note that since $Q$ is finite this
	process can only be done a finite number of times. Indeed, after having
	discarded at most $|Q| - 1$ concrete plays (which are disjoint after
	some finite point) it must be the case the last remaining possible
	concrete play has the desired property or we would have a contradiction
	with our assumptions.  Thus, there is some concrete play $\pi \in
	\gamma(\psi)$ that merges with infinitely many violating paths.
\qed\end{proof}

We now construct the non-deterministic B\"uchi automaton $\mathcal{N}'$ that
recognizes plays that contain a concrete path that merges with infinitely many
violating paths. The idea is that we non-deterministically keep track of two
paths: one that will eventually witness a violation and then merge with the
other, which ultimately serves as the witness for the path that merges with
infinitely many violating paths. When the two paths merge, the automaton
non-deterministically selects a new path to witness the violation. This is
achieved by non-deterministically selecting any state in the belief set of Eve
as these states represent the end states of any concrete play consistent with
the abstract play so far. To avoid the double exponential associated with
taking the Reif construction before determinizing the automaton, we instead
compute the belief set on-the-fly using a Moore machine that feeds into our
non-deterministic automaton. By transferring the exponential state increase to
an exponential increase in the alphabet size, the overall size of the
determinized automaton (after composition with the Moore machine) will be at most
singly exponential in the size of our game and $W$.

More specifically, denote by $\mathcal{B}$ the machine that, given
$\psi = o_0 \sigma_0 o_1 \sigma_1 \dots \in \plays(G)$ as its input yields the
infinite sequence $o_0 \sigma_0 s_0 o_1 \sigma_1 s_1 \dots \in (Obs \cdot
\Sigma \cdot \pow(Q))^\omega$ such that $s_0 = \{q_I\}$ and for all $i \ge 0$ we
have $s_{i+1} = \post_{\sigma_i}(s_i)$. One can easily give a definition of
$\mathcal{B}$ -- which closely resembles a subset construction -- with a state
space at most exponential w.r.t. $G$. Observe that the machine realizes a
continuous function, in the sense that every prefix of length $i$ of the input
uniquely defines the next $s_{i+1}$ annotation. Thus, the annotation can be done
on-the-fly.

Formally, $\mathcal{N}'$ consists of the state space $F' = Q \times Q \times
\{1,\dots,\lmax\} \times \{-W \cdot \lmax, \dots, -1 \} \cup \{\bot,\top\}$;
initial state $(q_I,q_I,1,\bot)$; input alphabet $\Sigma'' = \Sigma \times Obs
\times \pow(Q)$; and $\Delta''' \subseteq F \times \Sigma'' \times F$. The
transition relation $\Delta'''$ has a transition $((p,p',i,n), (\sigma,o,s),
(q,q',j,m))$ if $(p', \sigma, q') \in \Delta$, $q \in s$, $q' \in o$,
\[
	m =
	\begin{cases}
		w(p,\sigma,q) & \mbox{if } (p,\sigma,q) \in \Delta \land w(p,\sigma,q) < 0 \\
		n + w(p,\sigma,q) & \mbox{if } (p,\sigma,q)\in \Delta \land n \neq \bot \land n +
			w(p,\sigma,q) < 0 \land i < \lmax \\
		\top & \mbox{if } (p,\sigma,q) \in \Delta \land (n \neq \top
			\lor p \neq p') \land n \neq \bot \land i = \lmax \\
		\bot & \mbox{otherwise,}
	\end{cases}
\]
and
\[
	j = 
	\begin{cases}
		\lmax & \mbox{if } m = \top \\
		i + 1 & \mbox{if } m = n + w(p,\sigma,q) \\
		1 & \mbox{otherwise.}
	\end{cases}
\]
We say a state $(q,q',i,n) \in F'$ is accepting if $q = q'$, $n = \top$.  The
automaton accepts a word $x$ if and only it has a run $(q_0,q'_0,i_0,n_0)$
$(\sigma_0,o_1,s_1)$ $(q_1,q'_1,i_1,n_1)$ $(\sigma_1,o_2,s_2) \dots$ on $x$ such
that for infinitely many $j$ we have that $(q_j,q'_j,i_j,n_j)$ is accepting. 

\begin{proposition}
	\label{pro:ufix-lang}
	The non-deterministic B\"{u}chi automaton $\mathcal{N}'$ accepts a word
	$\alpha = \mathcal{B}(\psi)$, where $\psi \in \plays(G)$, if and
	only if $\psi \not\in \ufix(\lmax)$.
\end{proposition}

We recall that determinizing $\mathcal{N}'$ and complementing it yields an
exponentially bigger deterministic automaton. Its composition with
$\mathcal{B}$, itself exponentially bigger, accepts the desired set of plays and
is still singly exponential in the size of the original arena and $W$. Once
more, the desired result follows from the algorithm presented
in~\cite{jurdzinski00}.

\begin{theorem}
	\label{thm:ufix-exp-complete}
	Given WGA $G$, determining if \eve has a
	winning strategy for the $\ufix(\lmax)$~objective can be decided in time
	exponential in $W$ and the size of $G$. 
\end{theorem}

\begin{corollary}
	Given WGA $G$ with unary encoded weights, deciding if \eve has a
	winning strategy for the $\ufix(\lmax)$~objective is \EXP-complete.
\end{corollary}

\bibliographystyle{abbrv}
\bibliography{refs}

\newpage

\appendix

\section{Proof of Theorem~\ref{thm:mpfm-then-bnd}}
\begin{proof}
	In~\cite{ddgrt10} the authors show that if \eve is only allowed to play
	finite memory strategies then she wins the $\mpobjinf(\nu)$~game if and
	only if she wins the $\mpobjsup(\nu)$~game, for any $\nu \in
	\mathbb{Q}$. We show the claim holds for $\mpobjinf(\epsilon)$. Let
	$\lambda_\exists = \langle M, m_0, \alpha_u, \alpha_o \rangle$ be the
	deterministic Moore machine representation of \eve's finite memory
	winning strategy. Consider the product of the arena with \eve's finite
	memory winning strategy, $G \times M$, constructed in the obvious
	manner, i.e. every path in $G \times M$ corresponds to a concrete path
	consistent with her strategy.  Clearly all cycles in $G \times M$ have
	weight of at least $\epsilon$, otherwise \adam can create a concrete
	path with mean-payoff value less than $\epsilon$ by ``pumping'' the
	cycles with value less than $\epsilon$. As any path in $G \times M$
	corresponds to concrete plays consistent with \eve's strategy, this
	contradicts the fact that the strategy is winning for her. By Pigeonhole
	Principle we have that for any path in $G \times M$: if a window opens
	at step $i$, then after $i$ there is a sequence of length at most
	$|M||Q| - 1$ that is not involved in any cycle. Now, since every cycle
	has weight $\epsilon > 0$, after at most 
	\[
		\mu = \frac{W \cdot |M||Q|}{\epsilon} \cdot |M||Q|
	\]
	steps the window will have closed. It follows that for
	all $\psi \in \plays(G)$ consistent with her strategy:
	\[
		\forall \pi \in \gamma(\psi), \forall i \ge 0 : \pi \in \gw(i, \mu)
	\]
	which concludes our argument.
\qed\end{proof}

\section{Proof of Theorem~\ref{thm:bwmp-undec}}
\begin{proof}
	We provide a reduction from the universality of weighted finite automata
	which is undecidable~\cite{abk11}.  A \emph{weighted finite automaton}
	is a tuple $\mathcal{N}=\langle Q, \Sigma, q_I, \Delta, w, F \rangle$
	where $F \subseteq Q$ is a set of final accepting states. A {\em
	(accepting) run} of the automaton on a word $x = \sigma_0 \sigma_1 \dots
	\sigma_n \in \Sigma^*$ is a sequence $r = q_0 q_1 \dots q_n \in Q^+$
	such that $q_n \in F$ and $(q_i, \sigma_i, q_{i+1}) \in \Delta$ for all
	$0 \le i < n$. The cost of the run $r$ is $w(r) = \sum^{n-1}_{i=0}
	w(q_i,\sigma_i,q_{i+1})$. If the automaton is non-deterministic, it may
	have several runs on $x$. In that case, the \emph{cost of $x$ in
	$\mathcal{N}$} (denoted by $\mathcal{N}(x)$) is defined as the
	minimum of the costs of all its accepting runs on $x$.

	The \emph{universality} problem for weighted automata is to decide
	whether, for a given automaton $\mathcal{N}$, the following holds:
	\[
		\forall x \in \Sigma^* : \mathcal{N}(x) < 0.
	\]

	We construct a blind WGA, $G_\mathcal{N}$, so that:
	\begin{itemize}
		\item if $\mathcal{N}$ is universal, then \eve has a winning
			strategy for the objective \udirbnd,
		\item if $\mathcal{N}$ is not universal, then \adam has a
			winning strategy for the complement of the objective
			\bnd. 
	\end{itemize}
	As shown in Fig.~\ref{fig:relations}, $\udirbnd \subseteq \bnd$ and all
	the other BWMP objectives lie in between those two. So, our reduction
	establishes the undecidability of all BWMP objectives at once.
	
	Our reduction follows the gadgets given in
	Fig.~\ref{fig:infty-hash}-\ref{fig:no-asc-chains}-\ref{fig:sim-automaton}.
	When the game starts, \adam chooses to play from one of the three
	gadgets. As the game is blind for \eve, she does not know what is the
	choice of \adam and so she must be prepared for all possibilities. Note
	also that as \eve is blind, her strategy can be formalized by an
	infinite word $w \in \Sigma \cup \{ \# \}^{\omega}$.  Let us show first
	that the two first gadgets force \eve to play a word $w$ such that:
	\begin{itemize}
		\item[$(C_1)$] there are infinitely many $\#$ in $w$, and
		\item[$(C_2)$] there exists a bound $b \in \mathbb{N}$ such that
			the distance between two consecutive $\#$ in $w$ is
			bounded by $b$.
	\end{itemize}
	
	Assume that \eve plays a word $w=\#w_1 \# w_2\# w_3\# \dots \# w_n \#
	\dots$ that respects conditions $C_1$ and $C_2$, with each $w_i \in
	\Sigma^*$. 
	First, if \adam decides to play in the first gadget
	(Fig.~\ref{fig:infty-hash}), then either \adam stays in state $q_1$
	forever, and he does not open any window, or he decides at some
	point to go from $q_1$ to $q_2$, whereupon he does open a window.
	However, after at most $b$ steps \adam has to leave $q_2$ for $q_3$ at
	the next occurence of the $\#$ symbol, the bound $b$ is guaranteed by
	$C_2$.  After at most $b$ additional steps, the open window will be
	closed as the self loop on $q_3$ is labeled with the weight $+1$. So in
	this case, \eve wins the objective \udirbnd.
	Second, if \adam decides to play in the second gadget
	(Fig.~\ref{fig:no-asc-chains}), then he can go from $q_4$ to $q_5$ on
	the $\#$ symbols. The windows that open on those transitions will all
	close within $b$ steps according to condition $C_2$ and the game moves
	back to $q_4$. So again, \eve wins for the objective \udirbnd.

	Now assume that \eve plays a word $w$ that violates either condition
	$C_1$ or condition $C_2$.
	First, if $w$ violates $C_1$, then \adam chooses the first gadget
	(Fig.~\ref{fig:infty-hash}), and just after \eve has played her last
	$\#$,  \adam moves from $q_1$ to $q_2$. As there will be no $\#$
	anymore, \adam can loop on $q_2$ and the window that he has opened will
	never close. Hence, \adam wins for the complement of the objective \bnd.
	Second, if $w$ violates $C_2$ then there exists an infinite sequence of
	indices $i_1 < i_2 < \dots < i_n < \dots$ such that $|w_{i_1}| <
	|w_{i_2}| < \dots < |w_{i_n}| < \dots$. Then \adam can read this sequence
	of sub-words using runs of the form $q_4 (q_5)^* q_4$. Each such run
	will open a window that closes at the end of the sub-word. But as the
	sequence of lengths of the sub-words is strictly increasing and
	infinite, \adam wins for the complement of the objective \bnd.
	
	Now, we will assume that \eve plays a word $w=\#w_1 \# w_2\# w_3\# \dots
	\# w_n \# \dots$ that respects conditions $C_1$ and $C_2$, and we
	consider what happens when \adam plays in the third gadget
	(Fig.~\ref{fig:sim-automaton}). As we will see, in this gadget, the
	best strategy for \adam is to simulate the accepting runs of
	$\mathcal{N}$.
	
	Assume first that automaton $\mathcal{N}$ is {\em non-universal}. Then
	by definition, there exists a finite word $w_1 \in \Sigma^*$ such that
	all accepting runs of $\mathcal{N}$ on $w_1$ have a non-negative value,
	i.e. $\mathcal{N}(w_1) \geq 0$. In that case, $w=(\#w_1)^{\omega}$ is a
	finite memory winning strategy for \eve for the objective \bnd. Indeed,
	if \adam simulates accepting runs on $w$ then the mean-payoff of the
	outcome is at least $\frac{0.5}{b} > 0$ as each new $\#$ brings
	$+\frac{1}{2}$ and we know that $\mathcal{N}(w_1)\geq0$. So \eve wins
	for the objective \udirbnd\/ by Theorem~\ref{thm:mpfm-then-bnd}, as \eve
	obtains a strictly positive mean-payoff bounded away from zero with a
	finite memory strategy. Now, if \adam chooses not to follow accepting
	runs then the game ends up in state $\bot$ from which the mean-payoff
	is equal to $1$, so we can conclude using similar arguments that \eve
	wins for the objective \udirbnd. Thus in all these cases, \eve wins for
	the objective \udirbnd.
	
	Finally, assume that automaton $\mathcal{N}$ is universal and let us
	show then that \adam has a winning strategy for the complement of the
	\bnd\/ objective. Indeed, if \eve plays a word $w=\#w_1 \# w_2\# w_3\#
	\dots \# w_n \# \dots$ that respects conditions $C_1$ and $C_2$, then we
	know that $\mathcal{N}(w_i)<0$ for each $i \leq 0$. On such word, \adam
	can follow accepting runs in the gadget of Fig.~\ref{fig:sim-automaton}. As
	the length between two consecutive $\#$ is at most $b$, we know that the
	mean-payoff of the run constructed by \adam is less than or equal to
	$\frac{-0.5}{b}$. It follows that \adam wins the complement of the \bnd\/
	objective as claimed, as \bnd\/ objective implies the mean-payoff
	objectives (as shown in Fig.~\ref{fig:relations} ).
\qed\end{proof}

\section{Missing proofs from Section~\ref{sec:dirfix-games}}

\sizeofF*
\begin{proof}
\begin{align*}
	|\mathcal{F}| &\le (W \cdot \lmax)^{|Q| \cdot \lmax} \\
	&= \Big(2^{\log(W \cdot \lmax)}\Big)^{|Q| \cdot \lmax} \\
	&= 2^{|Q| \cdot \lmax \cdot \log(W \cdot \lmax)}.
\end{align*}
\qed\end{proof}

\suppisreachable*
\begin{proof}
	($\Rightarrow$) We proceed by induction. We will show that for all $0
	\le j \le n$, for all $q_j \in \supp(f_j)$ there is a concrete path $q_0
	\sigma_0 \dots q_j$ such that $q_k \in o_k$ for all $1 \le k \le j$ and
	$q_0 \in \supp(\phi)$. Note that for $j = 0$ the claim trivially holds.
	Assume the claim holds for $j$. From the definition of
	$\sigma$-successor and $\supp^{-1}$ we have that $\supp(f_{j+1}) =
	\post_{\sigma_{j}}(\supp(f_j)) \subseteq o_{j+1}$. This means that for
	all $q_{j+1} \in \supp(f_{j+1})$ there must be some $q_j \in \supp(f_j)$
	such that $(q_j, \sigma_j, q_{j+1}) \in \Delta$. Hence any $q_{j+1}$ is
	reachable from some $q_j$ via $\sigma_j$ which, by inductive hypothesis,
	is in turn reachable from some $q_0 \in \supp(\phi)$ via a concrete path
	of the desired form.

	($\Leftarrow$) We now show -- once more by induction on $j$ -- that for
	all $0 \le j \le n$, if there is a concrete path $q_0 \sigma_0 \dots
	q_j$ such that $q_0 \in \supp(\phi)$ and $q_k \in o_k$ for all $1 \le k
	\le j$, then $q_j \in \supp(f_j)$. The claim holds for $j = 0$. Assume
	that it holds for some $j$. From the assumptions we have that
	$(q_j, \sigma_j, q_{j+1}) \in \Delta$ and $q_{j+1} \in o_{k+1}$.
	Further, we know that $q_j \in \supp(f_j)$ by inductive hypothesis.
	Hence, $q_{j+1} \in \post_{\sigma_{j}}(\supp(f_j)) \subseteq o_{j+1}$
	which means that $q_{j+1} \in \supp(f_{j+1})$.
\qed\end{proof}

\thekey*

Instead of directly providing a proof of Lemma~\ref{lem:the-key}, we prove a
more general result below. Consider the three conditions stated in
Lemma~\ref{lem:magic-lemma}. We shall prove that that C1 $\Rightarrow$ C2
$\Rightarrow$ C3 $\Rightarrow$ C1. Since C1 corresponds to having a window of
length $l$ open at $p$ from $\supp(\phi)$, the desired result follows from
transitivity.

\begin{lemma}\label{lem:magic-lemma}
	Let $\rho = o_0 \sigma_0 \dots o_n$ be an abstract path, $\phi \in
	\mathcal{F}$ such that $\supp(\phi) \subseteq o_0$ and $\supp^{-1}(\rho,
	\phi) = f_0 \sigma_0 \dots f_n \in (\mathcal{F} \cdot \Sigma)^*$.
	Given state $p \in \supp(f_n)$ and $1 \le l \le \lmax$ such that $l \le
	n$, let $\lambda = n - l$. The following three statements are
	equivalent.
	\begin{enumerate}[C1.]
		\item There is a concrete path $q_0 \sigma_0 \dots q_n \in
			\gamma(\rho)$ with $q_n = p$ and $q_0 \in \supp(\phi)$
			and
			\[
				\sum_{j = n - l}^{m}
				w(q_j,\sigma_j,q_{j+1}) < 0
			\]
			for all $n - l \le m < n$.
		\item $f_n(p)_l < 0$.
		\item There is a concrete path $q_0 \sigma_0 \dots q_n 
			\in \gamma(\rho)$ with
			$q_n = p$ and $q_0 \in \supp(\phi)$ such that 
		\begin{enumerate}[(a)]
			\item $f_{j}(q_{j})_{j - \lambda}  < 0$
				for all $\lambda < j \le n$, and
			\item $f_{k}(q_{k})_{j - \lambda} +
				w(q_k,\sigma_k,q_{k+1}) = f_{k+1}(q_{k+1})_{k -
				\lambda + 1}$ for all $\lambda < k < n$.
		\end{enumerate}
	\end{enumerate}
\end{lemma}
\begin{proof}
	(C3 $\Rightarrow$ C1)
	We will apply induction on $m$. From the definition of
	$\sigma$-successor we have that $f_{\lambda + 1}(q_{\lambda + 1})_1 =
	\min\{0, w(q_\lambda, \sigma_\lambda, q_{\lambda +
	1})\}$. From assumption $(a)$ we know that $f_{\lambda + 1}(q_{\lambda
	+ 1})_1 < 0$. Thus, the claim holds for $m = \lambda$. Assume it holds
	for $m$. To conclude the proof, we now show that the claim holds for
	$m+1$ as well.
	\begin{align*}
		\sum_{j = n - l}^{m + 1} w(q_j,\sigma_j,q_{j+1})
			&= f_m(q_m)_{m-\lambda} + w(q_m,\sigma_m,q_{m+1}) &
			\text{ind.  hyp.}\\
		&= f_m(q_m)_{m-\lambda + 1} & \text{from $(b)$}\\
		&< 0 & \text{from $(a)$.}
	\end{align*}

	(C1 $\Rightarrow$ C2)
	We show, by induction on $m$, that for all $\lambda \le m < n$
	\[
		f_{m+1}(q_{m_1})_{m - \lambda + 1} \le \sum_{j = \lambda}^{m}
		w(q_j,\sigma_j,q_{j+1}).
	\]
	The desired result follows. As the base case, consider $m = \lambda$ and
	note that by definition of $\sigma$-successor we have that
	\begin{align*}
		f_{\lambda+1}(q_{\lambda+1})_1 &= \min(\{0\} \cup
			\{w(p,\sigma_{\lambda},q_{\lambda+1}) \st p \in
			\supp(f_\lambda) \land (p,\sigma_{\lambda},q_{\lambda + 1})
			\in \Delta\}) \\
		&\le w(q_\lambda,\sigma_\lambda,q_{\lambda+1}).
	\end{align*}
	Thus the claim holds. Assume that the claim is true for $m$. From the
	definition of $\sigma$-successor we have that
	\[
		f_{m+2}(q_{m+2})_{\lambda - m + 2} 
		\le f_{m+1}(q_{m+1})_{m - \lambda + 1} + w(q_{m+1},
		\sigma_{m+1}, q_{m+2}).
	\]
	From the inductive hypothesis we get have that the right hand side of
	the inequality is equivalent to
	\[
		\sum_{j = \lambda}^{m + 1} w(q_j,\sigma_j,q_{j+1}).
	\]
	Thus the claim holds for $m+1$ as well.

	(C2 $\Rightarrow$ C3)
	We inductively construct a concrete path $q_0 \sigma_0 \dots q_n
	\gamma(\rho)$ with $q_n = p$ and $q_0 \in \supp(\phi)$ such that
	\begin{enumerate}[(1)] 
		\item $f_{n - k}(q_{n - k})_{l - k} < 0$ for all $0 \le k < l$,
			and
		\item $f_{n-k}(q_{n-k})_{l-k} = w(q_{n-k+1},\sigma_{n-k+1},
			q_{n-k}) + f_{n-k+1}(q_{n-k+1})_{l-k+1}$ for all $1 \le
			k < l$.
	\end{enumerate}
	As these conditions are equivalent to $(a)$-$(b)$ from C3, the result
	follows. Note that for $k = 0$ we have that $(1)$ holds trivially since
	$p \in \supp(f_n)$ and $f_n(p)_l < 0$ by hypothesis.  If
	$f_{n-k}(q_{n-k})_{l-k} < 0$ then, by definition of $\sigma$-successor,
	it follows that there is some $q' \in \supp(f_{n-k+1}) \subseteq
	o_{n-k+1}$ such that $f_{n-k+1}(q')_{l-k+1} < 0$ and
	$f_{n-k}(q_{n-k})_{l-k} = w(q_{n-k+1},\sigma_{n-k+1}, q_{n-k}) +
	f_{n-k+1}(q_{n-k+1})_{l-k+1}$. In other words, $q'$ is the source of the
	minimal $\sigma_{n-k+1}$-transition of a state from $\supp(f_{n-k+1})$
	to $q_{n-k}$. Let $q_{n-k+1} = q'$. Continue in this fashion defining
	every $q_i$ up to $q_{n-l}$. Now, from
	Lemma~\ref{lem:supp-is-reachable}, we have that $q_{n-l}$ is reachable
	from some state in $\supp(\phi)$ via a concrete path of the desired
	form.  Any such path is a valid prefix for the sequence $q_{n-l}
	\sigma_{n-l} \dots q_n$ we constructed above.
\qed\end{proof}

\upreisuc*
\begin{proof}
	We have that for all $\sigma$, there is $h_\sigma \in S$ such that $(f,
	\sigma, h_\sigma) \in \Delta'$. By construction of $\Delta'$ we also
	know that there is $i_\sigma$ such that $(g, \sigma, i_\sigma) \in
	\Delta'$, and furthermore, since $\supp(f) \subseteq \supp(g)$, we get that
	\begin{align*}
		\supp(h_\sigma) &= \post_\sigma(\supp(f)) \cap o \\
		&\subseteq \post_\sigma(\supp(g)) \cap o \\
		&= \supp(i_\sigma)
	\end{align*}
	for some $o \in Obs$. Note that:
	\begin{enumerate}[(1)]
		\item since $f,g \not\in \calU$, then $f(p)_\lmax =
			g(p)_\lmax = 0$ for all $p \in \supp(f)$; and
		\item $i_\sigma(q)_1 = h_\sigma(q)_1$ for all $q \in
			\supp(h_\sigma)$.
	\end{enumerate}
	From $(1)$ and since $f \preceq g$, there is a function $\alpha :
	\{1,\dots,\lmax\} \to \{1,\dots,\lmax - 1\}$ such that for all $1 \le
	x < \lmax$ we have that $\alpha(x) \ge x$ and $f(p)_x \ge
	g(p)_{\alpha(x)}$ holds for all $p \in \supp(f)$. Observe that for all
	$q \in \supp(h_\sigma)$ and any $2 \le x \le \lmax$, we have that
	\begin{align*}
		h_\sigma(q)_x &= \min\limits_{p \in \supp(f)}(\{0\} \cup \{
			f(p)_{x-1} + w(p,\sigma,q) \st f(p)_{x-1} < 0\} \\
		&\ge \min\limits_{p \in \supp(f)}(\{0\} \cup \{
			g(p)_{\alpha(x-1)} + w(p,\sigma,q) \st
			g(p)_{\alpha(x-1)} < 0\} \\
		&\ge i_\sigma(q)_{\alpha(x-1)+1}.
	\end{align*}
	It follows that $h_\sigma \preceq i_\sigma$ and that, since $S$ is
	upward-closed, $i_\sigma \in S$. Thus, we have shown that for all
	$\sigma$, there is $i_\sigma \in S$ such that $(g, \sigma, i_\sigma) \in
	\Delta'$, which implies that $g \in \upre(S)$.
\qed\end{proof}

\section{Proof of Theorem~\ref{thm:ac-dirfix}}

\begin{lemma}
	\label{lem:upre-ac}
	Given upward-closed set $S \in \pow(\mathcal{F})$, $\lfloor \upre \rfloor
	(\lfloor S \rfloor) = \lfloor \upre (S) \setminus \calU \rfloor$.
\end{lemma}

\begin{proof}
	We first show that if $f \in \lfloor \upre \rfloor (\lfloor S \rfloor)$
	then $f \in \upre (S) \setminus \calU$. We have that $f \not\in \calU$
	and $\forall \sigma \in \Sigma, \exists q' \in \lfloor S \rfloor,
	\exists r'_\sigma \in \mathcal{F} : (f, \sigma, r'_\sigma) \in \Delta'
	\text{ and } q' \preceq r'_\sigma$. Since $S$ is upward-closed and $q'
	\preceq r'_\sigma$, we know that $r'_\sigma \in S$. Hence, we get that
	$\forall \sigma \in \Sigma, \exists r'_\sigma \in S : (f, \sigma,
	r'_\sigma) \in \Delta'$, which implies that $f \in \upre(S) \setminus
	\calU$.

	Next, we show that if $f \in \lfloor \upre(S) \setminus \calU \rfloor$
	then $f \in \{p' \in \mathcal{F} \setminus \calU \st \forall \sigma \in
	\Sigma, \exists q' \in a, \exists r' \in \mathcal{F} : (p', \sigma, r')
	\in \Delta' \text{ and } q' \preceq r'\}$. We know that $f \not\in
	\calU$ and $\forall \sigma \in \Sigma, \exists r' \in Q : (f, \sigma,
	r') \in \Delta'$. By definition of $\lfloor S \rfloor$, we know there is
	$q_{r'} \in \lfloor S \rfloor$ such that $q_{r'} \preceq r'$. Thus, we
	get that $\forall \sigma \in \Sigma, \exists r'_\sigma \in S, \exists
	q_{r'} \in \lfloor S \rfloor : (f, \sigma, r') \in \Delta'$ and $q_{r'}
	\preceq r'$.

	Finally, we note that if $f \in \lfloor \upre \rfloor (\lfloor S \rfloor)$
	then not only is it true that $f \in \upre (S) \setminus \calU$, but
	furthermore $f \in \lfloor \upre (S) \setminus \calU \rfloor$. Indeed,
	if this were not the case, then there would be $g \in \lfloor \upre (S)
	\setminus \calU \rfloor$ such that $g \preceq f$ and $f \neq g$. Then,
	by the argument explained in the previous paragraph, this would
	contradict minimality of $f$ in $\lfloor \upre \rfloor (\lfloor S
	\rfloor)$.  Similarly, if $f \in \lfloor \upre(S) \setminus \calU
	\rfloor$ then $f \in \lfloor \upre \rfloor (\lfloor S \rfloor)$, as
	otherwise, by the argument from the first paragraph of the proof,
	minimality in the first set would be contradicted. Thus, the claim
	holds.
\qed\end{proof}

\acalgodirfix*

\begin{proof}
	We note that for any upward-closed set $S \subseteq \mathcal{F}$ such
	that $\calU \subseteq S$ we have, from Lemma~\ref{lem:upre-is-uc} that
	$\calU \cup \upre(S)$ is again upward-closed and a superset of $\calU$.
	In fact, it holds that \begin{align*} \calU \cup \upre(S) &= (\lfloor
		\calU \rfloor \sqcup \lfloor \upre(S) \rfloor) \upclose & \\ &=
		(\lfloor \calU \rfloor \sqcup \lfloor \upre \rfloor (\lfloor S
		\rfloor) \upclose & \text{ from Lemma~\ref{lem:upre-ac}}. \\
	\end{align*}

	It is easy to show by induction that $\mu X . (\calU \cup \upre(X)) =
	\big(\mu X. (\lfloor \calU \rfloor \sqcup \lfloor \upre \rfloor (\lfloor
	X \rfloor))\big) \upclose$. Thus, $\{q_I'\} \not\sqsupseteq \mu X.
	(\lfloor \calU \rfloor \sqcup \lfloor \upre \rfloor(\lfloor S \rfloor))$
	if and only if $q_I \not\in \mu X.(\calU \cup \upre(S))$. From
	Proposition~\ref{pro:win-lose-safe} and Lemmas~\ref{lem:dirfix-cor}
	and~\ref{lem:dirfix-comp} we know this is the case if and only if \eve
	has a winning strategy in the safety game in $G'$ if and only if she
	wins the $\dirfix(\lmax)$~objective in $G$.
\qed\end{proof}

\section{Proof of Proposition~\ref{pro:fix-lang}}
\begin{proof}
	($\Rightarrow$) Assume $\mathcal{N}$ accepts $\psi$. Let $r = (q_0, i_0,
	n_0) (\sigma_0,o_1) (q_1, i_1, n_1) (\sigma_1, o_2) \dots$ be one of
	the accepting runs of the automaton on $\psi$. By construction of
	$\mathcal{N}$ we have that $q_0 \sigma_0 q_1 \sigma_1 \dots \in
	\gamma(\psi)$. Let $\pi_r$ denote this concrete play and $J = \{ j_0,
	j_1, j_2, \dots \}$ be an infinite set of indices such that $j_k <
	j_{k+1}$ and  $(q_{j_k}, i_{j_k}, n_{j_k})$ is accepting for all $k \ge
	0$. Such a sequence is guaranteed to exist since $r$ is accepting. One
	can easily verify by induction on the definition of $\Delta''$ that
	for all $k \ge 0$ it holds that $\pi_r \not\in \gw(i_{j_k} - \lmax,
	\lmax)$.
	It follows that $\forall m \ge 0, \exists n \ge m : \pi_r \not\in \gw(n,
	\lmax)$, which concludes our argument.

	($\Leftarrow$) Assume that $\psi = o_0 \sigma_0 o_1 \sigma_1
	\dots \not\in \fix(\lmax)$. Let $\pi = q_0 \sigma_0 q_1 \sigma_1 \in
	\gamma(\psi)$ be the concrete play such that for infinitely many $i$ it
	is the case that $\pi \not\in \gw(i, \lmax)$. We describe the infinite
	run of $\mathcal{N}$ on $\psi$ that accepts. Let $J = \{ j_0, j_1, j_2,
	\dots \}$ be an infinite set of indices such that $j_k + \lmax <
	j_{k+1}$ and $\pi \not\in \gw(j_k, \lmax)$ for $k \ge 0$. The sequence is
	guaranteed to exist because of our choice of $\pi$. Observe that this
	implies there is a run $r = (q_0, i_0, n_0)$ $(\sigma_0,o_1)$ $(q_1, i_1,
	n_1)$ $(\sigma_1, o_2) \dots$ of the automaton where for all $k \ge 0$ we
	have that $n_{j_k + 1} = w(q_{j_k},\sigma_{j_k},q_{j_k + 1})$ and for
	all $1 < l < \lmax$ then 
	\[
		n_{j_k + l} = n_{j_k + l - 1} + w(q_{j_k + l},\sigma_{j_k +
		l},q_{j_k + l + 1}).
	\]
	Furthermore, in this run it holds that for all $k \ge 0$ we have
	$i_{j_k + \lmax} = \lmax$. Hence, said run is such that for all $k \ge
	0$ the state $(q_{j_k + \lmax}, i_{j_k + \lmax}, n_{j_k + \lmax})$ is
	accepting. We conclude that the automaton accepts $\psi$.
\qed\end{proof}

\section{Proof of Proposition~\ref{pro:ufix-lang}}
\begin{proof}
	($\Rightarrow$) Assume $\mathcal{N}'$ accepts $\alpha$. Let $r = (q_0,
	q'_0 i_0, n_0)$ $(\sigma_0,o_1,s_1)$ $(q_1, q'_1 i_1, n_1)$ $(\sigma_1, o_2,
	s_2) \dots$ be one of the accepting runs of the automaton on $\alpha$.
	By construction of $\mathcal{N}'$ we have that $q'_0 \sigma_0 q'_1
	\sigma_1 \dots \in \gamma(\psi)$. Let $\pi_r$ denote this concrete
	play, $J = \{ j_0, j_1, \dots \}$ and $K = \{ k_0, k_1, \dots \}$ be
	two infinite sets of indices such that $j_l < j_{l+1}$ and $j_l + \lmax
	\le k_l$, for all $l \ge 0$, and for which we know that
	\begin{itemize}
		\item $(q_{k_l},q'_{k_l}, i_{k_l}, n_{k_l})$ is accepting for
			all $l \ge 0$, and
		\item $n_{j_l + \lmax} < 0 \land i_{j_l + \lmax} = \lmax$.
	\end{itemize}
	Such sequences are guaranteed to exist since $r$ is accepting. Assuming
	the correctness of $\mathcal{B}$, one can easily verify by induction on
	the definition of $\Delta'''$ that for all $l \ge 0$ we have that
	$\pi_r$, at $k_l$ merges with a path having a violation at $j_l$. It
	follows that $\pi_r$ merges with infinitely many violating paths. From
	Lemma~\ref{lem:ufix-eq-infopenwin} we get that $\psi \not\in
	\ufix(\lmax)$.

	($\Leftarrow$) Assume that $\psi = o_0 \sigma_0 o_1 \sigma_1
	\dots \not\in \ufix(\lmax)$. Let $\pi = q_0 \sigma_0 q_1 \sigma_1 \in
	\gamma(\psi)$ be the concrete play that merges with infinitely many
	violating paths (see Lemma~\ref{lem:ufix-eq-infopenwin}. We describe the
	infinite run of $\mathcal{N}$ on $\alpha = \mathcal{B}(\psi)$ that
	accepts. Let $J = \{ j_0, j_1, \dots \}$ and $K = \{ k_0, k_1, \dots
	\}$ be two infinite sets of indices such that $j_l < j_{l+1}$ and $j_l +
	\lmax + 1 < k_l$, for all $l \ge 0$, and for which we know that there is
	some $\chi_l \in \gamma(\psi[..k_l])$ such that $\pi[k_l] = \chi_l[k_l]$
	and for all $j_k < m \le j_k + \lmax + 1$ we have $w(\chi[j_k..m]) < 0$.
	The sequences are guaranteed to exist because of our choice of $\pi$.
	Observe that this implies there is a run $r = (q_0,q'_0 i_0, n_0)
	(\sigma_0,o_1,s_1) (q_1, q'_1 i_1, n_1) (\sigma_1, o_2,s_2) \dots$ of the
	automaton where for all $l \ge 0$ we have that $n_{j_l + 1} =
	w(q_{j_l},\sigma_{j_l},q_{j_l + 1})$ and for all $1 < b < \lmax$ then 
	\[
		n_{j_l + b} = n_{j_l + b - 1} + w(q_{j_l + b},\sigma_{j_l +
		b},q_{j_l + b + 1}).
	\]
	Furthermore, we have that $n_{j_l + b} = \top$ for all $\lmax \le b \le
	k_l$ and $q_{k_l} = q'_{k_l}$.  Hence, said run is such that for all $l
	\ge 0$ the state $(q_{k_l},q'_{k_l}, i_{k_l}, n_{k_l})$ is accepting. We
	conclude that the automaton accepts $\psi$.
\qed\end{proof}

\end{document}